\documentclass[onecolumn,draftclsnofoot,12pt,twoside]{IEEEtran}

\newlength\fwidth

\usepackage{pgfplots}
\usepackage{url}
\usepackage{times}
\usepackage{multirow}
\usepackage{amsfonts}
\usepackage{epsfig}
\usepackage{amsmath}
\usepackage{amssymb}
\usepackage[nolist]{acronym}
\usepackage[english]{babel}
\usepackage{cite}
\usepackage{color}
\usepackage{stfloats}
\usepackage{psfrag}
\usepackage{algorithm}
\usepackage{algorithmic}
\usepackage{subfigure}
\usepackage{graphicx}
\usepackage{amsthm}
\usepackage{bbm}	
\usepackage{graphicx}
  \pgfplotsset{compat=newest}
  \usetikzlibrary{plotmarks}
  \usetikzlibrary{arrows.meta}
  \usepgfplotslibrary{patchplots}
  \usepackage{grffile}

\usepackage{times}

\usepackage{latexsym}

\usepackage{bm}

\usepackage{ifpdf}
\ifpdf
  \usepackage{thumbpdf}
\else
\fi
\usepackage{fancyhdr}   

\usepackage{stfloats}
\usepackage{epsfig}
\usepackage{epstopdf}
\usepackage[justification=centering]{caption}
\usepackage{graphicx,subfigure}
\usepackage[keeplastbox]{flushend}

\usepackage{linguex}

\newtheorem{theorem}{Theorem}

\newtheorem{corol}{Corollary}

\IEEEoverridecommandlockouts
\let\svthefootnote\thefootnote
\begin{document}
\begin{acronym}
\acro{GOP}{group of pictures}
\acro{MSE}{mean squared error}
\acro{RDR}{rate-distortion region}
\acro{CSI}{channel state information}
\acro{FOV}{field of view}
\acro{i.i.d.}{independent and identically distributed}
\acro{r.v.}{random variable}
\acro{rr.vv.}{random variables}
\acro{DPCM}{differential predictive coded modulation}
\acro{AWGN}{additive white Gaussian noise}
\acro{SSTCG}{spatially memoryless, spatially stationary and temporally correlated Gaussian}
\end{acronym}

\IEEEoverridecommandlockouts

\setlength{\textfloatsep}{10pt plus 1.0pt minus 2.0pt}
\setlength{\floatsep}{10pt plus 1.0pt minus 2.0pt}
\setlength{\intextsep}{10pt plus 1.0pt minus 2.0pt}
\let\svthefootnote\thefootnote

\title{The Sum-Rate-Distortion Region of Correlated Gauss-Markov Sources}
\author{Giuseppe Cocco {\em Member, IEEE}\thanks{Giuseppe Cocco is with the LTS4 Signal Processing Laboratory and with the Laboratory of Intelligent Systems,
\'{E}cole Polytechnique F\'{e}d\'{e}rale de Lausanne, Lausanne, Switzerland, Email: giuseppe.cocco@epfl.ch}, and
Laura Toni {\em Member, IEEE}
\thanks{Laura Toni is with the Department of Electronic and Electrical Engineering, University College London,
London, UK, 
Email: l.toni@ucl.ac.uk}\thanks{Giuseppe Cocco is partly founded by the European Union's Horizon 2020 research and innovation programme under the Marie Sk{\l{}}odowska-Curie grant agreement No 751062}
}

\maketitle

\begin{abstract}
Efficient low-delay video encoders are of fundamental importance to provide timely feedback in remotely controlled platforms such as drones. In order to fully understand the theoretical limits of low-delay video encoders, we consider an ideal differential predictive coded modulation (DPCM) encoder and provide the explicit derivation of the sum-rate-distortion region for a generic number of successive correlated Gauss-Markov sources along the line of the work by Ma and Ishwar. 
Furthermore, we provide an upper bound on the minimum distortion achievable in case an arbitrary number of sources are not available at the decoder.
\end{abstract}
\IEEEpeerreviewmaketitle
\section{Introduction}
The widespread diffusion of consumer cameras such as those mounted on drones or skiers' helmets streaming videos in real-time is posing novel challenges in terms of bandwidth usage. In remotely controlled platforms such as drones the video received at the ground station is used as a feedback to steer the aircraft. This imposes stringent requirements in terms of delay. A \ac{DPCM} encoder coding frames on-the-fly with an IPPPP structure may help to decrease the video compression delay, at the price of compression efficiency reduction with respect to an encoder that jointly processes all frames in the \ac{GOP}. The need for highly efficient encoders that can meet stringent delay constraints requires a deep understanding of the theoretical limits in the compression of correlated sources, such as the consecutive frames in a video.

On this regard, in \cite{ViswBergerTIT2000} the sum-rate-distortion region for two \ac{SSTCG} sources with a \ac{MSE} distortion metric is derived. Such result is extended to three sources in \cite{ishwarTIT2015}, without a generalization to a generic number $M$ of sources. Spatially correlated sources have also been extensively studied in literature (see \cite{bergerRDT,tatikonda2004,stavrouControlLetters2018} and references therein). Although related, such source model is different from the one studied in \cite{ishwarTIT2015}. While the model in \cite{tatikonda2004} assumes correlation between adjacent source samples (spatial correlation), in \cite{ViswBergerTIT2000, ishwarTIT2015}  a correlation between symbols in the same position of consecutive source vectors (temporal correlation) is assumed. In real videos, both spatial and temporal correlations are present but modelling them in an accurate and yet mathematically tractable way is a challenging task. Furthermore, a full understanding of \ac{SSTCG} sources both with and without frames losses has not yet been achieved and is undergoing intense research.

In the following, we report the full derivation of the sum-rate-distortion region for a generic number of successive correlated Gauss-Markov sources \cite{ishwarTIT2015}. Khina \emph{et al.} recently published in  \cite{khina_ITW2017_seqCodGauss} the full characterization of the distortion-rate region for a generic number of Gauss-Markov sources. However, in \cite{khina_ITW2017_seqCodGauss} the sum-rate-distortion region is not explicitely calculated. The derivation presented here goes along the line of that in \cite{ViswBergerTIT2000,ishwarTIT2015}. Starting from this result, we derive the minimum distortion achievable by a $k$-step predictor, i.e., the minimum distortion achievable in the reconstruction of a source in case the previous $k$  sources  are not available, in case Gaussian descriptions are used. Such result is relevant for real-time video streaming over unreliable channels, in that it bounds the quality of the best reconstruction achievable by a source decoder when a generic number of consecutive frames in a \ac{GOP} are lost during transmission.
\section{System Model}
We consider a camera system acquiring, compressing and streaming video in real-time. The camera acquires video frames at a rate of $F_r$ frames per second\footnote{In the following we use the terms \emph{frame} and \emph{source} interchangeably.}. A lossy compressor is applied to the captured frames, generating GOPs of $M$ frames each and with an IPPPP structure having one reference (I) frame, followed by $M-1$ predicted (P) frames. Each P-frame depends only on the previous frame. Each frame is encoded  within $T_A=1/F_r$ seconds which is smaller than the GOP duration. In this way, the compressed frame can be transmitted before the successive frames in the \ac{GOP} have been acquired, thus reducing the latency with respect to an encoder that jointly compresses the whole \ac{GOP}. We start by considering a lossless communication channel and then we move to the case in which some frames are erased on the channel.

\subsection{Source Model}\label{sec:sourceModel}
The source model considered in the following is an \ac{SSTCG} process \cite{ViswBergerTIT2000}. In an \ac{SSTCG} source the intensity of a pixel generated by the source is correlated with the value of the same pixel in other time instants (frames) but independent of the values of other pixels in the same or in other time instants. Let $n$ be the number of pixels in the source image.
A new frame is generated by the source every inter-frame period, i.e., every $T_f=1/F_r$ seconds. The $t$-th generated frame is an $n$-dimensional vector, {that can be seen as the vectorization of a bi-dimensional $\sqrt{n}\times\sqrt{n}$ matrix}, which we indicate as
$\mathbf{X}_t = \left(X_t(1),X_t(2),\ldots,X_t(n-1),X_t(n)\right)$.
$\mathbf{X}_t$ is a vector of \ac{i.i.d.} zero-mean Gaussian variables having variance $\sigma^2_t$. 
The intensity of the pixels in consecutive frames corresponding to a given point in the scene is modelled as a temporal Markov process\footnote{A triplet of discrete random variables $X, Y, Z$ forms a Markov chain in that order (denoted $X-Y-Z$) if their joint probability mass function satisfies $p(x,y,z)=p(x)p(y|x)p(z|y)$\cite{coverThomas}. The definition extends in a similar way to the case of continuous random variables.}, i.e., $\forall t, t>1$ we have
\begin{eqnarray}\label{eqn:markov}
 X_{t-1}(i)-X_{t}(i)-X_{t+1}(i).
\end{eqnarray}
Such model has been widely used for raw videos \cite{ViswBergerTIT2000,yangITA2009} and for the evolution of the
innovations process along optical-flow motion trajectories for groups of adjacent pixels \cite{ishwarTIT2015}.

\subsection{Source Encoder}\label{sec:souenc}
Given a frame $\mathbf{X}_t$, the source encoder generates a compressed version  that can be described with the least number of bits per symbol while satisfying a constraint on the error (\emph{distortion}) between the corresponding reconstruction $\widehat{\mathbf{X}}_t$ and $\mathbf{X}_t$\cite{bergerRDT}. 
We consider a per-frame \ac{MSE} average distortion metric. Specifically, let us define the following:
\begin{equation}\label{eqn:distortion}
d^{(n)}_t\left(\mathbf{X}_t,\widehat{\mathbf{X}}_t\right)\triangleq \frac{1}{n}\sum_{i=1}^n\left({X}_t(i)-\widehat{{X}}_t(i)\right)^2.
\end{equation}
The average distortion is defined as $\mathbf{E}\left\{d^{(n)}_t\left(\mathbf{X}_t,\widehat{\mathbf{X}}_t\right)\right\}$, where the average is taken with respect to the distribution of the source vectors. We define the target distortion tuple $\mathbf{D}=\left(D_1, D_2, D_3,\ldots\right)$\footnote{In practical system it is common practice to set a common $D$ for the whole \ac{GOP}, which is a special case of the model we consider here.}. It is required that, for large $n$, the average distortion for frame number $t$ is lower than or equal to $D_t$, i.e., 
\begin{equation}\label{eqn:dist_tuple}
\lim_{n\rightarrow \infty}\mathbf{E}\left\{d^{(n)}_t\left(\mathbf{X}_t,\widehat{\mathbf{X}}_t\right)\right\}\leq D_t.
\end{equation}
In the following we assume $D_t<\sigma_t^2$, $\forall t\in\{1,2,\ldots,M\}$.
The source encoder we consider is an idealized \ac{DPCM} encoder. Such source encoder has been shown in \cite{ishwarTIT2015} to be optimal for the considered source model, in the sense that it achieves the minimum sum-rate when an \ac{MSE} distortion measure is adopted for all distortion values within the rate-distortion region. The idealized \ac{DPCM} encoder works as follows.
When the first frame of a GOP $\mathbf{X}_1$ is captured and made available at the source encoder, it is compressed using a minimum sum-rate source codebook at a rate $R_{1}(D_1)$ bits per source symbol. The encoding of the first frame in each GOP is done independently of all previous frames. Once the encoding of the first frame is completed, the index of the description $\widehat{\mathbf{U}}_1$ of the corresponding source codeword $\widehat{\mathbf{X}}_1$ is sent over the channel. When, after $T_f$ seconds, the second frame is generated,  the source encoder compresses it taking into account ${\mathbf{X}}_1$, ${\mathbf{X}}_2$ and $\widehat{\mathbf{X}}_1$ and outputs the auxiliary vector $\widehat{\mathbf{U}}_2$ having rate $R_{2}(D_2)$ bits per symbol. In general, the $t$-th frame in a GOP is compressed taking into account all available frames ${\mathbf{X}}^t={\mathbf{X}}_1,\ldots,\mathbf{X}_t$ and all available encoder outputs $\widehat{\mathbf{U}}^{t-1}=\widehat{\mathbf{U}}_1,\widehat{\mathbf{U}}_2,\ldots,\widehat{\mathbf{U}}_{t-1}$.
Frames are source-encoded in groups of $M$, where $M$ is the product between the GOP duration expressed in seconds and the frame rate $F_r$ expressed in Hz. This models an IPPPP video compressor in which a given frame within a GOP can be reconstructed only if all previous source-coded frames of the same GOP are available at the decoder. 
\subsection{Source Decoder}
The decoder at time $t$ generates a reconstruction $\widehat{X}_t$ of frame $X_t$ using all available encoder outputs received so far $\widehat{\mathbf{U}}^{t}=\widehat{\mathbf{U}}_1,\widehat{\mathbf{U}}_2,\ldots,\widehat{\mathbf{U}}_{t-1}, \widehat{\mathbf{U}}_{t}$ and trying to achieve the desired distortion tuple $\mathbf{D}=\left(D_1, D_2, \ldots, D_t\right)$.
\section{Sum-Rate Distortion Region}
The source coding scheme described in Section \ref{sec:souenc} is similar to the one proposed in \cite{ViswBergerTIT2000}, in which two correlated source vectors are successively generated and encoded. In \cite{ishwarTIT2015} the \ac{RDR} for a generic number of frames with generic encoding and decoding delays is derived. The sum-rate-distortion region is also derived and the results are specialized for the case of three correlated Gaussian sources ($M=3$). In the following theorem the approach of \cite{ishwarTIT2015} is used to explicitly calculate the sum-rate-distortion region for a generic number $M$ of Gaussian source vectors.
\begin{theorem}\label{theorem1}
The minimum sum-rate within the rate-distortion region for $M$ successive correlated Gauss$-$Markov sources and MSE distortion is
\begin{eqnarray}
R^{(M)}_{\Sigma}(\mathbf{D})=\sum_{i=1}^M\frac{1}{2}\log^+\left(\frac{\sigma^2_{W_i}}{D_i}\right)
\end{eqnarray}
where
\begin{equation}
    \sigma^2_{W_t}= 
\begin{cases}
    \sigma^2_1,& \text{for } t= 1\\
    \rho_s^2\frac{\sigma^2_{t}}{\sigma_{t-1}^2}D_{t-1}+(1-\rho_s^2)\sigma^2_{t},              & \text{for } t> 1.
\end{cases}
\end{equation}
and $\log^+(x)=\max(0,\log(x))$.
\end{theorem}
$\log(.)$ being the base 2 logarithm.
\begin{proof}
We start by finding an upper bound to the minimum sum-rate within the rate-distortion region. Then we derive a lower bound and show that the two coincide.
\subsubsection{Upper Bound}
Consider $M$ Gaussian sources $X_1,\ldots,X_M$ such that 
\begin{eqnarray}\label{eqn:proof1}
 X_{t-1}-X_{t}-X_{t+1},
\end{eqnarray}
$\forall t\in\{2,\ldots,M-1\}$.
We can write:
\begin{eqnarray}\label{eqn:proof2}
 X_{t}=\rho_s \frac{\sigma_t}{\sigma_{t-1}}X_{t-1} + N_t, \forall t>1,
\end{eqnarray}
where $\rho_s $ is the correlation coefficient between symbols in the same position of two consecutive source words, $N_t\sim\mathcal{N}(0,(1-\rho_s^2)\sigma_t^2)$ is independent of $X_{t-1}$ and represents the innovation of $X_t$ with respect to $X_{t-1}$.
Let us consider the first source $X_1$. Since $\sigma_1^2\geq D_1$, according to the test channel model \cite[Chapter 10]{coverThomas}, it is possible to find two mutually independent random variables $\widehat{X_{1}}\sim\mathcal{N}(0,\sigma_1^2-D_1)$ and $Z_{1}\sim\mathcal{N}(0,D_1)$ such that 
\begin{eqnarray}\label{eqn:proof3}
X_{1}=\widehat{X_{1}} + Z_{1}.
\end{eqnarray}
$\widehat{X_{1}}$ represents the source-encoded version of $X_1$ after reconstruction and, since an ideal quantizer is assumed, it can approximate $X_1$ with a distortion $D_1$ using a rate 
\begin{eqnarray}
R_{1}(D_1)=\frac{1}{2}\log^+\left(\frac{\sigma^2_{1}}{D_1}\right).
\end{eqnarray}
The next source (frame) $X_2$ can be expressed as:
\begin{eqnarray}\label{eqn:proof4}
 X_{2}=\rho_s \frac{\sigma_2}{\sigma_1}X_{1} + N_2 = \rho_s \frac{\sigma_2}{\sigma_1}\widehat{X_{1}} + W_2,
\end{eqnarray}
where 
\begin{eqnarray}\label{eqn:proof5}
W_2=\rho_s \frac{\sigma_2}{\sigma_1}Z_1 + N_2,
\end{eqnarray}
is a zero-mean Gaussian \ac{r.v.} with variance $\sigma_{W_2}^2=\rho_s^2\frac{\sigma_2^2}{\sigma_1^2}D_1+(1-\rho_s^2)\sigma_2^2$. The source encoder encodes $W_2$ using an ideal quantizer, generating the compressed version $\widehat{W}_2$, such that $W_2=\widehat{W}_2+Z_2$, where $\widehat{W_{2}}\sim\mathcal{N}(0,\sigma_{W_2}^2-D_2)$ and $Z_{2}\sim\mathcal{N}(0,D_2)$ are independent. Since an ideal quantizer is assumed, $\widehat{W}_2$ can be described using a rate equal to $R_2=1/2\log^+\left({\sigma^2_{W_2}}/{D_2}\right)$.
At this point we note that 
\begin{eqnarray}\label{eqn:proof7}
 \widehat{X}_{2}= \rho_s \frac{\sigma_2}{\sigma_1}\widehat{X_{1}} + \widehat{W}_2,
\end{eqnarray}
 achieves the distortion $D_2$ with rate $R_2$, since $E\left\{\left(X_2-\widehat{X}_2\right)^2\right\}=E\{\left(Z_2\right)^2\}=D_2$ by construction.
 Now we apply induction to show that the same procedure can be iterated obtaining the desired distortion at the desired rate for all successive sources.
 Let us assume that source $X_t$ has been successfully encoded with rate $R_t={1}/{2}\log^+\left({\sigma^2_{W_t}}/{D_t}\right)$ so that a reconstructed version $\widehat{X}_t$ achieving distortion $D_t$ can be obtained. From Eqn. \ref{eqn:proof2} and Eqn. \ref{eqn:proof4} we have
 \begin{eqnarray}\label{eqn:proof8}
 X_{t+1}=\rho_s \frac{\sigma_{t+1}}{\sigma_t}X_{t} + N_{t+1} = \rho_s \frac{\sigma_{t+1}}{\sigma_t}\widehat{X_{t}} + W_{t+1},
\end{eqnarray}
where $W_{t+1}=\rho_s \frac{\sigma_{t+1}}{\sigma_t}Z_t + N_{t+1}$
is a zero-mean Gaussian \ac{r.v.} with variance $\sigma_{W_{t+1}}^2=\rho_s^2\frac{\sigma_{t+1}^2}{\sigma_t^2}D_t+(1-\rho_s^2)\sigma_{t+1}^2$. Using an ideal quantizer it is possible to find two independent random variables $\widehat{W}_{t+1}\sim\mathcal{N}(0,\sigma_{W_{t+1}}^2-D_{t+1})$ and $Z_{t+1}\sim\mathcal{N}(0,D_{t+1})$ such that $W_{t+1}=\widehat{W}_{t+1}+Z_{t+1}$ and $\widehat{W}_{t+1}$ reproduces ${W}_{t+1}$ with distortion $D_{t+1}$ using a rate
\begin{eqnarray}\label{eqn:proof10}
R_{t+1}=\frac{1}{2}\log^+\left(\frac{\sigma^2_{W_{t+1}}}{D_{t+1}}\right).
\end{eqnarray}
At this point it is sufficient to use 
\begin{eqnarray}\label{eqn:proof11}
\widehat{X}_{t+1}=\rho_s \frac{\sigma_{t+1}}{\sigma_t}\widehat{X}_t + \widehat{W}_{t+1},
\end{eqnarray}
as reproduction \ac{r.v.} for $X_{t+1}$. Since the decoder already knows $\widehat{X}_t$, the rate required to encode $\widehat{X}_{t+1}$ is the same as $\widehat{W}_{t+1}$. Finally, by construction we have $E\left\{\left(X_{t+1}-\widehat{X}_{t+1}\right)^2\right\}=E\{\left(Z_{t+1}\right)^2\}=D_{t+1}$.
Given the above, the following holds
\begin{eqnarray}\label{eqn:proof12}
R^{(M)}_{\Sigma}(\mathbf{D})\leq\sum_{i=1}^M\frac{1}{2}\log^+\left(\frac{\sigma^2_{W_i}}{D_i}\right).
\end{eqnarray}
\subsubsection{Lower Bound}
The sum-rate-distortion region for the successive transmission of sources forming a Gauss-Markov process is (setting $k=0$ in \cite[Corollary 5.1]{ishwarTIT2015}):
\begin{eqnarray}\label{eqn:proof13}
R^{(M)}_{\Sigma}(\mathbf{D})=\min I(X^M;\widehat{X}^M)
\end{eqnarray}
where the minimum of the mutual information\footnote{The mutual information between two continuous random variables $X$ and $Y$ is defined as $I(X,Y)=h(X)-h(X|Y)$, $h(X)$ being the differential entropy of variable $X$.} is taken over all distributions of $\widehat{X}^M$ satisfying the following:
\begin{eqnarray}
E[d_j(X_j,\widehat{X}^{j-1})]\leq D_j, j=1,\ldots,M \\ \label{eqn:proof14}
\widehat{X}_j-(X^j,\widehat{X}^{j-1})-X^M_{j+1}, j=1,\ldots,M-1. \label{eqn:proof15}
\end{eqnarray}
The following inequalities hold:
\begin{align}\label{eqn:proof16}
R^{(M)}_{\Sigma}(\mathbf{D})&=\min I(X^M;\widehat{X}^M)\notag\\
&\overset{(a)}{=}\min \sum_{i=1}^MI(X^M;\widehat{X}_i \text{\textbar} \widehat{X}^{i-1})\notag\\
&\overset{(b)}{=}\min \sum_{i=1}^MI(X^i;\widehat{X}_i \text{\textbar} \widehat{X}^{i-1})\notag\\
&\overset{(c)}{\geq}\min \sum_{i=1}^MI(X_i;\widehat{X}_i \text{\textbar} \widehat{X}^{i-1})\notag\\
&\overset{(d)}{=}\min \sum_{i=1}^M\left[h(X_i\text{\textbar}\widehat{X}^{i-1})-h(X_i\text{\textbar}\widehat{X}^i)\right]\notag\\
&\overset{}{\geq}h(X_1)+\min\left\{\sum_{i=1}^{M-1}\left[h(X_{i+1}\text{\textbar}\widehat{X}^i)\notag\right.\right.\\
&-\left.\left.h(X_{i}\text{\textbar}\widehat{X}^i)\right]-h(X_M-\widehat{X}_M)\right\}\notag\\
&\overset{(e)}{\geq}\frac{1}{2}\log^+(2\pi e\sigma_1^2)-\frac{1}{2}\log^+(2\pi eD_M)\notag\\
&+\sum_{i=1}^{M-1}\min\left[h(X_{i+1}\text{\textbar}\widehat{X}^i)-h(X_{i}\text{\textbar}\widehat{X}^i)\right],
\end{align}
where (a) follows from the chain rule for mutual information, (b) is because of \ref{eqn:proof15}, (c) follows again from the chain rule for mutual information, (d) is by definition of mutual information while (e) follows from the fact that the Gaussian distribution maximizes entropy. Now, note that the Markov chain $\widehat{X}^j-X_j-X_{j+1}$ holds for $j=1,\ldots,M-1$, which means that the $j+1$-th source is independent of all previous source reconstructions once conditioned to the $j$-th source. By using this in \cite[Lemma 5]{ViswBergerTIT2000} we can write:
\begin{eqnarray}\label{eqn:proof17}
h(X_{i+1}\text{\textbar}\widehat{X}^i)-h(X_{i}\text{\textbar}\widehat{X}^i)\geq \frac{1}{2}\log^+\left(\frac{\sigma_{W_{i+1}}^2}{D_i}\right).
\end{eqnarray}
By plugging \ref{eqn:proof17} into \ref{eqn:proof16} we obtain:
\begin{eqnarray}\label{eqn:proof18}
R^{(M)}_{\Sigma}(\mathbf{D})&\geq& \sum_{i=1}^{M}\frac{1}{2}\log^+\left(\frac{\sigma_{W_{i}}^2}{D_i}\right).
\end{eqnarray}
From equations \ref{eqn:proof12} and \ref{eqn:proof18} we see that the right-hand side of \ref{eqn:proof18} is both an upper and a lower bound for $R^{(M)}_{\Sigma}(\mathbf{D})$, which concludes the proof.
\end{proof}
If the communication between the source encoder and the source decoder takes place over an erasure channel, the reconstruction according to Theorem \ref{theorem1} is not possible. This is due to the ideal \ac{DPCM} encoder, which only encodes the difference between $X_t$ and its best approximation obtainable from $\widehat{X}_{t-1}$. In the following corollary, we derive the minimum distortion attainable by the source decoder in case the first $t-k$ frames are correctly received while the last $k$ frames are lost ($k$-step predictor) in case Gaussian descriptions are used. The loss can be due, for instance, to erasures on the channel.
\begin{corol}\label{corollary1}
Given $t$, $t>0$, successive correlated Gauss-Markov sources of which the first $t-k$, $0 \leq k<t$, are source encoded using a \ac{DPCM} within the \ac{RDR} for a given distortion tuple $\mathbf{D}=\left(D_1,\ldots,D_{t-k}\right)$ and the relative reconstructions $\widehat{X}_1,\ldots,\widehat{X}_{t-k}$ are available at the source decoder, the minimum distortion achievable for source $X_t$ in case Gaussian descriptors are used is:
\begin{equation}
\sigma^2_{W_{t,k}}=\rho_s^{2k}\frac{\sigma^2_{t}}{\sigma_{t-k}^2}D_{t-k}+\left(1-\rho_s^{2k}\right)\sigma^2_{t},
\end{equation}
where $D_0\triangleq 0.$
\end{corol}
\begin{proof} (Sketch)
Theorem \ref{theorem1} guarantees that $X_{t-k}$ can be reconstructed with a distortion less than or equal to $D_{t-k}$. Restricting ourselves to the case of Gaussian descriptions, from Eqn. \ref{eqn:proof8} it follows that: 
\begin{eqnarray}
 {X}_{t-k+1}&=&\rho_s \frac{\sigma_{t-k+1}}{\sigma_{t-k}}{X}_{t-k}+{N}_{t-k+1}\label{eqn:coroll1a}\\ 
 &=&\rho_s \frac{\sigma_{t-k+1}}{\sigma_{t-k}}\left(\widehat{X}_{t-k}+Z_{t-k}\right)+{N}_{t-k+1}\label{eqn:coroll1b}\\
 &=&\rho_s \frac{\sigma_{t-k+1}}{\sigma_{t-k}}\widehat{X}_{t-k}+\phi_{t-k+1}\label{eqn:coroll1c},
\end{eqnarray}
where we defined 
$$\phi_{t-k+1}\triangleq\rho_s \frac{\sigma_{t-k+1}}{\sigma_{t-k}}Z_{t-k}+{N}_{t-k+1},$$
 $Z_{t-k}$ being the reconstruction error relative to the last available reconstruction and ${N}_{t-k+1}\sim\mathcal{N}(0,(1-\rho_s^2)\sigma_{t-k+1}^2)$ is the innovation of source ${t-k+1}$ with respect to source ${t-k}$.
Since the decoder has knowledge of only $\widehat{X}_{t-k}$, the best reconstruction of ${X}_{t-k+1}$ it can generate is 
$\rho_s \frac{\sigma_{t-k+1}}{\sigma_{t-k}}\widehat{X}_{t-k}$
which, by construction, achieves a distortion equal to the variance of ${W}_{t-k+1}$, i.e., $\sigma_{W_{t-k+1}}^2$. Thus, according to \ref{eqn:coroll1c}, $\phi_{t-k+1}$ is the one-step reconstruction error.
Iterating Eqn. \ref{eqn:coroll1a} we obtain the following expression for the $k$-step reconstruction error $\phi_{t}$:
\begin{eqnarray}\label{eqn:coroll2}
\phi_{t}=\rho_s^k \frac{\sigma_{t}}{\sigma_{t-k}}Z_{t-k}+\sum_{j=1}^k\rho_s^{k-j}\frac{\sigma_{t}}{\sigma_{t-k+j}}{N}_{t-k+j}.
\end{eqnarray}
The \ac{MSE} error in the reconstruction of ${X}_t$ is $E\left\{\phi_{t}^2\right\}$. Since all random variables in \ref{eqn:coroll2} are zero-mean and independent, the \ac{MSE} is:
\begin{eqnarray}\label{eqn:coroll3}
\sigma^2_{W_{t,k}}&=&E\left\{\phi_{t}^2\right\}\notag\\
&=&\rho_s^{2k}\frac{\sigma^2_{t}}{\sigma_{t-k}^2}D_{t-k}+\sigma_{t}^2\sum_{j=1}^k\rho_s^{2(k-j)}\frac{{\sigma^2}_{t-k+j}}{\sigma_{t-k+j}^2}\left(1-\rho_s^2\right)\notag\\
&=&\rho_s^{2k}\frac{\sigma^2_{t}}{\sigma_{t-k}^2}D_{t-k}+\left(1-\rho_s^2\right)\sigma_{t}^2\sum_{j=0}^{k-1}\left(\rho_s^{2}\right)^j\notag\\
&=&\rho_s^{2k}\frac{\sigma^2_{t}}{\sigma_{t-k}^2}D_{t-k}+\left(1-\rho_s^2\right)\sigma_{t}^2\frac{1-\rho_s^{2k}}{1-\rho_s^2}\notag\\
&=&\rho_s^{2k}\frac{\sigma^2_{t}}{\sigma_{t-k}^2}D_{t-k}+\left(1-\rho_s^{2k}\right)\sigma_{t}^2\notag.
\end{eqnarray}
\end{proof}
As a final remark, we note that such distortion bounds from above the minimum \ac{MSE} achievable by a decoder that, based on the currently available reconstruction, tries to approximate a source which is $k$ steps ahead.
\section{Conclusion}\label{sec:conclusion}
We provided the derivation of the sum-rate-distortion region for a generic number of successive correlated Gauss-Markov sources, along the line of the result presented in \cite{ishwarTIT2015} for the case of 3 sources. Starting from this result, we derived the minimum distortion achievable for source number $t$ in case only the first $t-k$ sources' reconstructions are available at the decoder and Gaussian descriptions are used. Such result is relevant for real-time video streaming over wireless channels, because, for the considered model, it gives a bound on the quality of the best reconstruction achievable by a source decoder when a generic number $k$ of consecutive frames in a \ac{GOP} are lost on the channel.
\IEEEtriggeratref{0}
\bibliographystyle{IEEEtran}
\bibliography{ArXivSRDR}
\flushend
\end{document}